\documentclass[11pt,leqno]{article}

\usepackage{latexsym,times,graphicx,amsmath,amsfonts,amssymb,amsthm,bbm,natbib}

\newtheoremstyle{localthm}
	{5pt} % space above
	{5pt} % space below
	{\sl} % Body font
	{} % Indent amount
	{\bf} % Theorem head font
	{} % Punctuation after theorem head
	{.7em} % Space after theorem head
	{} % Theorem head spec ?

\newtheoremstyle{localrem}
	{5pt} % space above
	{5pt} % space below
	{\rm} % Body font
	{} % Indent amount
	{\bf} % Theorem head font
	{} % Punctuation after theorem head
	{.7em} % Space after theorem head
	{} % Theorem head spec ?

\theoremstyle{localthm}
\newtheorem{Theorem}{Theorem}[section]
\newtheorem{Lemma}[Theorem]{Lemma}
\newtheorem{Corollary}[Theorem]{Corollary}

\theoremstyle{localrem}

\def\bea{\begin{eqnarray*}}
\def\eea{\end{eqnarray*}}

\def\bs{\boldsymbol}

\def\Ab{\bs{A}}
\def\etab{\bs{\eta}}
\def\Gamb{\bs{\Gamma}}
\def\Mb{\bs{M}}
\def\Nb{\bs{N}}
\def\Sigb{\bs{\Sigma}}
\def\thb{\bs{\theta}}

\def\Wb{\bs{W}}
\def\xb{\bs{x}}

\def\Yb{\bs{Y}}
\def\zb{\bs{z}}
\def\Zb{\bs{Z}}

\def\B{\mathbb{B}}
\def\R{\mathbb{R}}

\def\Ex{\mathop{\mathrm{I\!E}}\nolimits}
\def\Pr{\mathop{\mathrm{I\!P}}\nolimits}

\def\Cov{\mathop{\mathrm{Cov}}\nolimits}
\def\Var{\mathop{\mathrm{Var}}\nolimits}

\def\argmin{\mathop{\mathrm{arg\,min}}}

\def\LL{\mathcal{L}}
\def\NN{\mathcal{N}}

\def\eps{\epsilon}

\begin{document}
%===============

\addtolength{\baselineskip}{+.27\baselineskip}

\thispagestyle{empty}

%\begin{center}
%	{\large University of Bern}\\
%	{\large Institute of Mathematical Statistics and Actuarial Science}
%	
%	\textbf{\large Technical Report 75}
%\end{center}

%\vfill

\title{Stochastic Search for\\Semiparametric Linear Regression Models}
\author{Lutz D\"{u}mbgen$^*$ (University of Bern)\\
	Richard J. Samworth$^\dagger$ (University of Cambridge)\\
	Dominic Schuhmacher$^*$ (University of Bern)}
\date{June 2011, revised October 2011}

\maketitle

\begin{abstract}
This paper introduces and analyzes a stochastic search method for parameter estimation in linear regression models in the spirit of \cite{BeranMillar1987}. The idea is to generate a random finite subset of a parameter space which will automatically contain points which are very close to an unknown true parameter. The motivation for this procedure comes from recent work of \cite{DSS2011} on regression models with log-concave error distributions.
\end{abstract}

\bigskip

\textbf{AMS (2000) subject classifications:} \
62G05, 62G09, 62G20, 62J05.

\textbf{${}^*$} \
Work supported by Swiss National Science Foundation.

\textbf{${}^\dagger$} \
Work supported by a Leverhulme Research Fellowship.

%\newpage

%=======================
\section{Introduction}
\label{sec:Introduction}
%=======================

This paper introduces and analyzes a stochastic search method for parameter estimation in linear regression models in the spirit of \cite{BeranMillar1987}. The idea is to generate a random finite subset of a parameter space which will automatically contain points which are very close to an unknown true parameter. The motivation for this procedure comes from recent work of \cite{DSS2011} on regression models with log-concave error distributions. Section~\ref{sec:Linear regression} reviews the latter setting.  In section~\ref{sec:Stochastic search} the stochastic search method is described and analyzed in detail. Our construction relies on the exchangeably weighted bootstrap as introduced by \cite{MasonNewton1992} and developed further by \cite{PraestgaardWellner1993}. While these papers are dealing with i.i.d.\ random elements, the present considerations will show that the exchangeably weighted bootstrap is also asymptotically valid in heteroscedastic linear regression models under mild regularity conditions. Thus it is a viable alternative to the wild bootstrap as proposed by \cite{Wu1986}. All proofs are deferred to section~\ref{sec:Proofs}.

%==============================================================
\section{Linear regression with log-concave error distribution}
\label{sec:Linear regression}
%==============================================================

Suppose that for integers $q$ and $n \ge q$ we observe $(\xb_{n1},Y_{n1}), (\xb_{n2},Y_{n2}), \ldots, (\xb_{nn},Y_{nn})$, where
$$
	Y_{ni} \ = \ \thb_n^\top \xb_{ni} + \eps_{ni}
$$
with an unknown parameter $\thb_n \in \R^q$, fixed design vectors $\xb_{n1}, \xb_{n2}, \ldots, \xb_{nn} \in \R^q$ and independent real random errors $\eps_{n1}, \eps_{n2}, \ldots, \eps_{nn}$ with mean zero. We assume that our regression model includes the constant functions, i.e.\ the column space of the design matrix $\bs{X}_n := [\xb_{n1}, \xb_{n2}, \ldots, \xb_{nn}]^\top$ contains the constant vector $(1)_{i=1}^n$.

\paragraph{Maximum likelihood estimation.}
Suppose that the errors $\eps_{ni}$ are identically distributed with density $f_n$ such that $\psi_n := \log f_n$ is concave. One may estimate $f_n$ and $\thb_n$ consistently via maximum likelihood as follows: Let $\Phi_o$ be the set of all concave functions $\phi : \R \to [-\infty,\infty)$ such that
$$
	\int_{\R} e_{}^{\phi(y)} \, dy \ = \ 1
	\quad\text{and}\quad
	\int_{\R} y e_{}^{\phi(y)} \, dy \ = \ 0 .
$$
Then we define $(\hat{\psi}_n, \hat{\thb}_n)$ to be a maximizer of
$$
		\sum_{i=1}^n \phi(Y_{ni} - \etab^\top \xb_{ni}) ,
$$
over all pairs $(\phi,\etab) \in \Phi_o \times \R^q$, provided such a maximizer exists. It follows indeed from \cite{DSS2011} that $(\hat{\psi}_n, \hat{\thb}_n)$ is well-defined almost surely if $n \geq q+1$. Precisely, the MLE exists whenever $\bs{Y}_n = (Y_{ni})_{i=1}^n$ is not contained in the column space of the design matrix $\bs{X}_n$.  Simulation results in \cite{DSS2011} indicate that $\hat{\thb}_n$ may perform substantially better than the ordinary least squares estimator, for instance when the errors have a skewed, log-concave density.

\paragraph{Consistency.}
General results of \cite{DSS2011} imply that the MLE is consistent in the following sense, where asymptotic statements refer to $n \to \infty$, unless stated otherwise:

\begin{Theorem}
\label{thm:Consistency}
Suppose that $q = q(n)$ such that $q(n)/n \to 0$ and
$$
	\int_{\R} \bigl| f_n(y) - f(y) \bigr| \, dy \ \to \ 0
$$
for some probability density $f$. Then $\hat{f}_n := \exp(\hat{\psi}_n)$ satisfies
$$
	\int_{\R} \bigl| \hat{f}_n(y) - f_n(y) \bigr| \, dy
	\ \to_p \ 0 ,
$$
and
$$
	\frac{1}{n} \sum_{i=1}^n \min \bigl( \bigl| (\hat{\thb}_n - \thb_n)^\top \xb_{ni} \bigr|, 1 \bigr)
	\ \to_p \ 0 .
$$
\end{Theorem}

For fixed dimension $q$ and under additional conditions on the design points $\xb_{ni}$, Theorem~\ref{thm:Consistency} implies a stronger consistency property of $\hat{\thb}_n$, where $\|\cdot\|$ denotes standard Euclidean norm, and $\lambda_{\rm min}(\bs{A})$ denotes the minimal eigenvalue of a symmetric matrix $\bs{A}$:

\begin{Corollary}
\label{cor:Consistency}
Suppose that the assumptions of Theorem~\ref{thm:Consistency} are satisfied, where the dimension $q$ is fixed. In addition, suppose that
$$
	\liminf_{n \to \infty} \, \lambda_{\rm min} \Bigl(
		\frac{1}{n} \sum_{i=1}^n \xb_{ni}^{}\xb_{ni}^\top \Bigr)
	\ > \ 0
$$
and
$$
	\lim_{n,c \to \infty} \, \frac{1}{n} \sum_{i=1}^n
		1 \bigl\{ \|\xb_{ni}\| > c \bigr\} \|\xb_{ni}\|^2
	\ = \ 0 .
$$
Then
$$
	\bigl\| \hat{\thb}_n - \thb_n \bigr\| \ \to_p \ 0 .
$$
\end{Corollary}

%=====================================
\section{Stochastic search}
\label{sec:Stochastic search}
%=====================================

Computing the MLE $(\hat{\psi}_n, \hat{\thb}_n)$ from the previous section is far from trivial. For any fixed $\etab \in \R^q$, the profile log-likelihood
$$
	L_n(\etab) \ := \ \max_{\phi \in \Phi_o}
		\sum_{i=1}^n \phi(Y_{ni} - \xb_{ni}^\top \etab)
$$
can be computed quickly by means of algorithms described by \cite{DR2011}. Furthermore, as shown by \cite{DSS2011}, $L_n(\cdot)$ is continuous and coercive in that $L_n(\etab) \to - \infty$ as $\|\etab\| \to \infty$. However, numerical examples reveal that $L_n(\cdot)$ is not concave or even unimodal in the sense that the sets $\bigl\{ \etab \in \R^q : L_n(\etab) \ge c \bigr\}$, $c \in \R$, are convex.

To deal with this problem, we resort to a stochastic search strategy in the spirit of \cite{BeranMillar1987}. In particular, we construct a random finite subset $\Theta_n$ of $\R^q$ such that
\begin{equation}
\label{eq:Asymptopia}
	\min_{\etab \in \Theta_n} \, n^{1/2} \|\etab - \thb_n\|
	\ \to_p \ 0 .
\end{equation}
Then we redefine $\hat{\thb}_n$ to be a maximizer of $L_n(\cdot)$ over $\Theta_n$. A close inspection of the proofs reveals that the consistency results in Theorem~\ref{thm:Consistency} and Corollary~\ref{cor:Consistency} carry over to this new version. Moreover, if the original MLE $\hat{\thb}_n$ satisfies $\|\hat{\thb}_n - \thb_n\| = O_p(n^{-1/2})$, which is an open conjecture of ours, the same would be true for the stochastic search version.

\paragraph{Exchangeably weighted bootstrap.}
For the remainder of this section we describe and analyze a particular construction of $\Theta_n$: Let $\Wb_n^{(1)}, \Wb_n^{(2)}, \Wb_n^{(3)}, \ldots$ be i.i.d.\ random weight vectors in $[0,\infty)^n$, independent from the data $(\bs{X}_n,\bs{Y}_n)$. Then we consider the ordinary least squares estimator
$$
	\check{\thb}_n^{(0)}
	\ := \ \argmin_{\etab \in \R^q}
		\sum_{i=1}^n (Y_{ni} - \xb_{ni}^\top \etab)^2
	\ = \ \Bigl( \sum_{i=1}^n \xb_{ni}^{} \xb_{ni}^\top \Bigr)^{-1}
		\sum_{i=1}^n Y_{ni}^{} \xb_{ni}^{}
$$
and the randomly weighted least squares estimators
$$
	\check{\thb}_n^{(b)}
	\ := \ \argmin_{\etab \in \R^q}
		\sum_{i=1}^n W_{ni}^{(b)} (Y_{ni} - \xb_{ni}^\top \etab)^2
	\ = \ \Bigl( \sum_{i=1}^n W_{ni}^{(b)} \xb_{ni}^{} \xb_{ni}^\top \Bigr)^{-1}
		\sum_{i=1}^n W_{ni}^{(b)} Y_{ni}^{} \xb_{ni}^{}
$$
for $b = 1, 2, 3, \ldots$, where $\Wb_n^{(b)} = (W_{ni}^{(b)})_{i=1}^n$. If $\sum_{i=1}^n \xb_{ni}^{}\xb_{ni}^\top$ or $\sum_{i=1}^n W_{ni}^{(b)} \xb_{ni}^{} \xb_{ni}^\top$ happens to be singular, we interpret its inverse as generalized inverse. If we define
$$
	\Theta_n \ := \ \bigl\{ \check{\thb}_n^{(b)} : 0 \le b \le B_n \bigr\}
$$
with integers $B_n \to \infty$, the subsequent considerations imply that \eqref{eq:Asymptopia} is satisfied under certain conditions.

\paragraph{Asymptotics.}
We assume that the random weight vectors $\Wb_n^{} := \Wb_n^{(b)}$ satisfy the following three conditions:\\[1ex]
(W.1) \ The random variables $W_{n1}, W_{n2}, \ldots, W_{nn}$ are exchangeable and satisfy
$$
	\sum_{i=1}^n W_{ni} \ \equiv \ n .
$$
(W.2) \ For a given number $c > 0$, \
$$
	\frac{1}{n} \sum_{i=1}^n (W_{ni} - 1)^2 \ \to_p \ c^2 .
$$
(W.3) \ As $n \to \infty$ and $K \to \infty$,
$$
	\frac{1}{n} \sum_{i=1}^n W_{ni}^2 1\{W_{ni}^{} \ge K\} \ \to_p \ 0 .
$$

Note that (W.1) implies that $\Ex W_{n1} = 1$. In fact, when (W.1) holds, conditions~(W.2-3) are a consequence of the following moment conditions:\\[1ex]
(W.4) \ For a given number $c > 0$, \
$$
	\Var(W_{n1}) \ \to \ c^2 .
$$
Moreover,
$$
	\limsup_{n \to \infty} \mathrm{Cov}(W_{n1}^2,W_{n2}^2) \ \le \ 0
	\quad\text{and}\quad
	\limsup_{n \to \infty} \Ex(W_{n1}^4) \ < \ \infty .
$$

To see this, observe that under (W.1) and (W.4),
\bea
	\Ex \biggl\{\frac{1}{n}\sum_{i=1}^n (W_{ni}-1)^2\biggr\}
	& = & \Var(W_{n1}) \ \to \ c^2
		\quad\text{and} \\
	\Var \biggl\{\frac{1}{n}\sum_{i=1}^n (W_{ni}-1)^2\biggr\}
	& = & \Var\biggl(\frac{1}{n}\sum_{i=1}^n W_{ni}^2\biggr) \\
	& = & \frac{1}{n} \Var(W_{n1}^2) + \frac{n-1}{n} \Cov(W_{n1},W_{n2}) \\
	& \le & \frac{1}{n} \Ex(W_{n1}^4) + o(1) \ \to \ 0
\eea
as $n \to \infty$, which proves (W.2). Moreover,
$$
	\Ex\biggl\{\frac{1}{n}\sum_{i=1}^n W_{ni}^2 1\{W_{ni}^{} \ge K\}\biggr\}
	\ = \ \Ex(W_{n1}^2 1\{W_{n1}^{} \ge K\})
	\ \le \ \Ex(W_{n1}^4)/K^2
	\ \to \ 0
$$
as $n \to \infty$ and $K \to \infty$. This proves (W.3).

As to the data $(\bs{X}_n,\bs{Y}_n)$, we drop the assumption of identically distributed errors and only require the $\eps_{ni}$ to have mean zero and finite variances. Further we assume that the following three conditions are satisfied:\\[1ex]
(D.1) \ For a fixed positive definite matrix $\bs{\Gamma} \in \R^{q\times q}$,
$$
	\frac{1}{n} \sum_{i=1}^n \xb_{ni}^{} \xb_{ni}^\top
		\ \to \ \bs{\Gamma} .
$$
(D.2) \ For a fixed matrix $\bs{\Gamma}_\eps \in \R^{q\times q}$,
$$
	\frac{1}{n} \sum_{i=1}^n \Var(\eps_{ni}) \xb_{ni}^{} \xb_{ni}^\top
		\ \to \ \bs{\Gamma}_\eps .
$$
(D.3) \ With $L_{ni} := n^{-1} (1 + \eps_{ni}^2) \|\xb_{ni}\|^2$,
$$
	\Ex \sum_{i=1}^n L_{ni} \min(L_{ni},1) \ \to \ 0 .
$$

Note that (D.1-2) implies that $\Ex \sum_{i=1}^n L_{ni} \to \mathrm{trace}(\Gamb + \Gamb_\eps)$.  Even under the weaker condition $\Ex \sum_{i=1}^n L_{ni} = O(1)$, condition~(D.3) is easily shown to be equivalent to the following Lindeberg-type condition:\\[1ex]
(D.3$^\prime$) \ For any fixed $\delta > 0$,
$$
	\Ex \sum_{i=1}^n L_{ni} 1\{L_{ni} > \delta\} \ \to \ 0 .
$$

\begin{Theorem}
\label{thm:Asymptopia}
Suppose that conditions (W.1-3) and (D.1-3) are satisfied.

\noindent
\textbf{(a)} \ For any fixed integer $B \ge 1$,
\bea
	n^{1/2} \bigl( \check{\thb}_n^{(b)} - \thb_n \bigr)_{b=0}^B
	& \to_{\LL}^{} &
	\Gamb^{-1} \bigl( \Zb^{(0)} + 1\{b \ge 1\} c \Zb^{(b)} \bigr)_{b=0}^B
\eea
with independent random vectors $\Zb^{(0)}, \Zb^{(1)}, \Zb^{(2)}, \ldots, \Zb^{(B)}$ having distribution $\mathcal{N}_q(\bs{0}, \bs{\Gamma}_\eps)$.

\noindent
\textbf{(b)} \ For arbitrary integers $B_n \to \infty$,
$$
	\min_{b=1,2,\ldots,B_n} n^{1/2} \bigl\| \check{\thb}_n^{(b)} - \thb_n^{} \bigr\|
	\ \to_p \ 0 .
$$
\end{Theorem}

Part~(a) of this theorem is illustrated in Figure~\ref{fig:Ellipses}. Asymptotically, $\check{\thb}_n^{(0)}$ (depicted as $\bullet$) behaves like $\thb_n$ (depicted as $\star$) plus $n^{-1/2} \Gamb^{-1} \Zb^{(0)}$. From the latter point one gets to $\check{\thb}_n^{(b)}$, $b \ge 1$, by adding another Gaussian random vector $c n^{-1/2} \Gamb^{-1} \Zb^{(b)}$. Writing $\LL(A)$ for the law of a random vector $A$, the ellipses with broken lines indicate $\LL \bigl( \check{\thb}_n^{(0)} \bigr)$, while the ellipses with solid lines indicate $\LL \bigl( \check{\thb}_n^{(b)} \,\big|\, \Yb_n \bigr)$.

\begin{figure}
\centering
\includegraphics[width=0.7\textwidth]{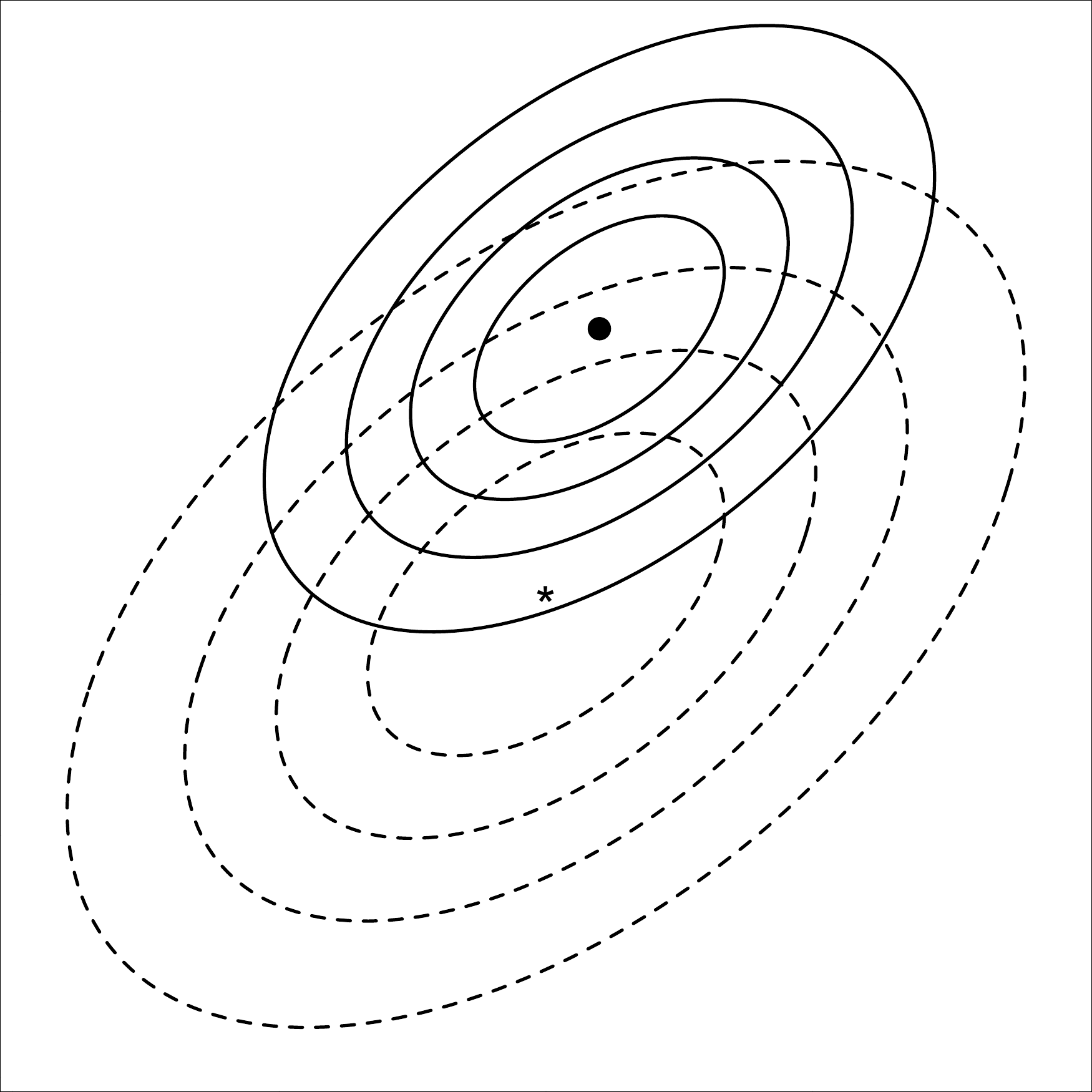}
\caption{Illustration of Theorem~\ref{thm:Asymptopia}~(a).}
\label{fig:Ellipses}
\end{figure}

Note that for any fixed integer $B \ge 1$,
$$
	\min_{b=1,\ldots,B} n^{1/2} \bigl\| \check{\thb}_n^{(b)} - \thb_n^{} \bigr\|
	\ \to_{\LL} \ 
	\min_{b=1,\ldots,B} \bigl\| \Gamb^{-1} (\Zb^{(0)} + c \Zb^{(b)}) \bigr\| .
$$
The next result provides a more detailed analysis of the latter random variable in terms of the way its distribution depends on $\Sigb := \Gamb^{-1} \Gamb_{\eps} \Gamb^{-1}$ and $B$. Recall that the Weibull distribution $\mathrm{Weibull}(q)$ is defined as the distribution on $[0,\infty)$ with distribution function $F(x) := 1 - \exp(- x^q)$.

\begin{Theorem}
\label{thm:Limit experiment}
Let $\Zb_1, \Zb_2, \Zb_3, \ldots$ be independent random vectors in $\R^d$ with continuous density~$f$. For any fixed $\zb \in \R^q$,
$$
	\alpha_q^{} \, f(\zb)^{1/q} \,
		B^{1/q} \, \min_{b=1,\ldots,B} \|\Zb_b - \zb\|
	\ \to_{\LL} \ \mathrm{Weibull}(q)
$$
as $B \to \infty$, where $\alpha_q := \pi^{1/2} \Gamma(q/2 + 1)^{-1/q}$.
\end{Theorem}

Presumably this result is well-known to people familiar with nearest neighbor methods, but for the reader's convenience a proof is given in section~\ref{sec:Proofs}. It implies the following result for our particular setting:

\begin{Corollary}
\label{cor:Limit experiment}
Let $\Zb_0, \Zb_1, \Zb_2, \ldots$ be independent random vectors with distribution $\NN_q(\bs{0},\Sigb)$, where $\Sigb$ is nonsingular. Then for any fixed $c > 0$,
$$
	B^{1/q} \, \min_{b=1,\ldots,B} \|\Zb_0 + c \Zb_b\|
	\ \to_{\LL} \ \beta_q^{} \,
	\det(\Sigb)_{}^{1/(2q)} \, c \, \exp \Bigl( \frac{S^2}{2c^2 q} \Bigr) W
$$
as $B \to \infty$ with independent random variables $S^2 \sim \chi_q^2$ and $W \sim \mathrm{Weibull}(q)$, where $\beta_q := 2^{1/2} \Gamma(q/2 + 1)^{1/q}$.
\end{Corollary}

If we drop the assumption that $\Sigb$ is nonsingular, the conclusion of Corollary~\ref{cor:Limit experiment} remains true if we replace $q$ with $\mathrm{rank}(\Sigb)$ and $\det(\Sigb)$ with the product of the nonzero eigenvalues of $\Sigb$.

The previous result shows that $\min_{b=1,\ldots,B} \|\Zb_0 + c \Zb_b\|$ is of order $O_p(B^{-1/q})$ as $B \to \infty$. This means, roughly speaking, that to achieve a small approximation error $\delta > 0$, one has to generate $O(\delta^{-q})$ points. This is coherent with the well-known fact that a Euclidean ball with fixed (large) radius can be covered with $O(\delta^{-q})$ but no less balls of radius $\delta$. However, note that the limiting distribution also depends on $\det(\Sigb)$. If we fix $\mathrm{trace}(\Sigb) = \Ex(\|\Zb_b\|^2)$ but decrease $\det(\Sigb)$, the asymptotic distribution of the minimal distance gets stochastically smaller.

For large dimension $q$, the stochastic factor $c \, \exp \bigl( S^2 / (2c^2 q) \bigr)$ in Corollary~\ref{cor:Limit experiment} can be approximated by $c \, \exp(1 / (2c^2))$, because $\Ex(S^2/q) = 1$ and $\Var(S^2/q) = 2/q$. Differentiation with respect to $c$ reveals that $c = 1$ is the unique minimizer of the latter approximate factor. Hence choosing $c = 1$ is approximately optimal in high dimensions. Alternatively, one could use $c = c_q := \mathrm{Median}(S^2/q)$.

\paragraph{Examples of weighting schemes.}
\cite{PraestgaardWellner1993} describe many different weighting schemes satisfying (W.1-3). Let us just recall two of them:

Sampling uniformly at random with replacement (the usual bootstrap sampling scheme) corresponds to a weight vector $\Wb_n$ with multinomial distribution $\mathrm{Multi}(n; \, n^{-1}, n^{-1}, \ldots, n^{-1})$. Here one can show that (W.1) and (W.4) are satisfied with $c = 1$.

Another interesting strategy is subsampling without replacement: For a fixed integer $m_n \in \{1,\ldots,n-1\}$ let $\Wb_n$ be a uniform random permutation of a vector with $m_n$ components equal to $n/m_n$ and $n - m_n$ components equal to zero. Then $\sum_{i=1}^n W_{ni} = n$ and
$n^{-1} \sum_{i=1}^n (W_{ni} - 1)^2 = n/m_n - 1$. The latter expression converges to $c^2$ if, and only if, $m_n/n \to (c^2 + 1)^{-1}$. In that case, $n^{-1} \sum_{i=1}^n W_{ni}^2 1\{W_{ni} > K\} = 0$ for sufficiently large $n$, provided that $K > c^2 + 1$. Note that $c = 1$ is achieved if $m_n/n \to 1/2$.

\paragraph{Asymptotic validity of the bootstrap.}
As a by-product of Theorem~\ref{thm:Asymptopia}, we obtain the asymptotic validity of the exchangeably weighted bootstrap for the case $c = 1$. Precisely, with $\Sigb = \Gamb^{-1} \Gamb_\eps \Gamb^{-1}$,
$$
	\LL \bigl( n^{1/2}(\check{\thb}_n^{(0)} - \thb_n) \bigr)
	\ \to_w \ \NN_q(\bs{0}, \Sigb)
$$
and
$$
	\LL \bigl( n^{1/2}(\check{\thb}_n^{(1)} - \check{\thb}_n^{(0)})
		\,\big|\, \Yb_n \bigr)
	\ \to_{w,p} \ \NN_q(\bs{0}, c^2 \Sigb) ,
$$
where $\to_{w,p}$ stands for weak convergence in probability. This latter assertion means that
$$
	G_n := \Ex \Bigl( g \bigl( n^{1/2}(\check{\thb}_n^{(1)} - \check{\thb}_n^{(0)}) \bigr)
		\,\Big|\, \Yb_n \Bigr)
	\ \to_p \ \Ex g(c \Zb)
$$
for any bounded and continuous function $g : \R^q \to \R$, where $\Zb \sim \NN_q(\bs{0},\Sigb)$. To verify this, we employ a trick of \cite{Hoeffding1952}: It follows from Theorem~\ref{thm:Asymptopia} that
$$
	n^{1/2} \bigl( \check{\thb}_n^{(1)} - \check{\thb}_n^{(0)},
		\check{\thb}_n^{(2)} - \check{\thb}_n^{(0)} \bigr)
$$
converges in distribution to $(c\Zb_1, c\Zb_2)$ with independent copies $\Zb_1, \Zb_2$ of $\Zb$. Hence
$$
	\Ex G_n
	\ = \ \Ex g \bigl( n^{1/2}(\check{\thb}_n^{(1)} - \check{\thb}_n^{(0)}) \bigr)
	\ \to \ \Ex g(c \Zb) .
$$
Furthermore, by independence of $\Wb_n^{(1)}$, $\Wb_n^{(2)}$ and $\Yb_n$,
\bea
	\Ex (G_n^2)
	& = & \Ex \Ex \Bigl(
		g \bigl( n^{1/2}(\check{\thb}_n^{(1)} - \check{\thb}_n^{(0)}) \bigr)
		\cdot
		g \bigl( n^{1/2}(\check{\thb}_n^{(2)} - \check{\thb}_n^{(0)}) \bigr)
		\,\Big|\, \Yb_n \Bigr) \\
	& = & \Ex \Bigl(
		g \bigl( n^{1/2}(\check{\thb}_n^{(1)} - \check{\thb}_n^{(0)}) \bigr)
		\cdot
		g \bigl( n^{1/2}(\check{\thb}_n^{(2)} - \check{\thb}_n^{(0)}) \bigr) \Bigr) \\
	& \to & \Ex \bigl( g(c\Zb_1) \cdot g(c\Zb_2) \bigr)
		\ = \ \bigl( \Ex g(c\Zb) \bigr)^2 ,
\eea
whence $\Var(G_n) \to 0$.

%=================
\section{Proofs}
\label{sec:Proofs}
%=================

\begin{proof}[\bf Proof of Corollary~\ref{cor:Consistency}]
It suffices to show that for any nonrandom sequence $(\etab_n)_n$ in $\R^q$,
$$
	\frac{1}{n} \sum_{i=1}^n \min \bigl( |\xb_{ni}^\top \etab_n|, 1 \bigr)
	\ \to \ 0
$$
implies that $\etab_n \to \bs{0}$. To this end we write $\etab_n = \|\etab_n\| \bs{u}_n$ with a unit vector $\bs{u}_n \in \R^q$. For any fixed number $\eps > 0$, it follows from $\|\etab_n\| \ge \eps$ that
\bea
	\frac{1}{n} \sum_{i=1}^n \min \bigl( |\xb_{ni}^\top \etab_n|, 1 \bigr)
	& \ge & \frac{1}{n} \sum_{i=1}^n \min \bigl( \eps^2 |\xb_{ni}^\top\bs{u}_n|^2, 1 \bigr) \\
	& \ge & \frac{1}{n} \sum_{i=1}^n \Bigl( \eps^2 (\xb_{ni}^\top\bs{u}_n)^2
		- 1 \bigl\{ \eps \|\xb_{ni}\| > 1 \bigr\} \eps^2 \|\xb_{ni}\|^2 \Bigr) \\
	& \ge & \eps^2 \Bigl( \lambda_{\rm min} \Bigl(
		\frac{1}{n} \sum_{i=1}^n \xb_{ni}^{}\xb_{ni}^\top \Bigr)
		- \frac{1}{n} \sum_{i=1}^n 1 \bigl\{ \|\xb_{ni}\| > 1/\eps \bigr\} \|\xb_{ni}\|^2 \Bigr) .
\eea
But the lower bound on the right hand side is bounded away from zero, provided that $\eps > 0$ is sufficiently small. This shows that $\|\etab_n\| < \eps$ for sufficiently large $n$.
\end{proof}

In our proof of Theorem~\ref{thm:Asymptopia} we make repeated use of the following elementary lemma:

\begin{Lemma}
\label{lem:Auxiliary}
For some $n \ge 2$ let $\bs{V} = (V_i)_{i=1}^n \in [0,\infty)^n$ and $\Mb_1, \Mb_2, \ldots, \Mb_n \in \R^d$ be independent random vectors, where $\bs{V}$ is a uniform random permutation of a fixed vector $\bs{v} = (v_i)_{i=1}^n \in [0,\infty)^n$ while $\Ex\|\Mb_i\| < \infty$ for $1 \le i \le n$. Then for an arbitrary constant $K \ge 0$,
$$
	\Ex \Bigl\| \sum_{i=1}^n V_i \Mb_i - \bar{v} \sum_{i=1}^n \Ex \Mb_i \Bigr\|
	\ \le \ 2 R(K) S + 2 \bar{v} L + \Bigl( \frac{n}{n-1} K \bar{v} L \Bigr)^{1/2}
$$
where $\bar{v} := n^{-1} \sum_{i=1}^n v_i$ and
\bea
	R(K) & := & \frac{1}{n} \sum_{i=1}^n v_i 1\{v_i > K\} , \\
	S & := & \sum_{i=1}^n \Ex \|\Mb_i\| , \\
	L & := & \sum_{i=1}^n \Ex \|\Mb_i\| \min(\|\Mb_i\|, 1) .
\eea
\end{Lemma}

\begin{proof}[\bf Proof of Lemma~\ref{lem:Auxiliary}]
Let $V_i' := V_i 1\{V_i > K\}$ and $W_i := V_i 1\{V_i \le K\}$. The corresponding means are $R(K)$ and $\bar{w} := n^{-1} \sum_{i=1}^n v_i 1\{v_i \le K\} \le \bar{v}$, respectively. Then
\bea
	\lefteqn{ \Ex \Bigl\| \sum_{i=1}^n V_i \Mb_i - \bar{v} \sum_{i=1}^n \Ex \Mb_i \Bigr\| } \\
	& \le & \Ex \Bigl\| \sum_{i=1}^n W_i \Mb_i - \bar{w} \sum_{i=1}^n \Ex \Mb_i \Bigr\|
		+ \sum_{i=1}^n \Ex V_i' \|\Mb_i\| + R(K) \Bigl\| \sum_{i=1}^n \Ex \Mb_i \Bigr\| \\
	& \le & \Ex \Bigl\| \sum_{i=1}^n W_i \Mb_i - \bar{w} \sum_{i=1}^n \Ex \Mb_i \Bigr\|
		+ 2 R(K) S .
\eea
Further, let $\Mb_i' := \min(\|\Mb_i\|,1) \Mb_i$ and $\Nb_i := (1 - \|\Mb_i\|)^+ \Mb_i$. Then
\bea
	\lefteqn{ \Ex \Bigl\| \sum_{i=1}^n W_i \Mb_i - \bar{w} \sum_{i=1}^n \Ex \Mb_i \Bigr\| } \\
	& \le & \Ex \Bigl\| \sum_{i=1}^n W_i \Nb_i - \bar{w} \sum_{i=1}^n \Ex \Nb_i \Bigr\|
		+ \sum_{i=1}^n \Ex W_i \|\Mb_i'\| + \bar{w} \Bigl\| \sum_{i=1}^n \Ex \Mb_i' \Bigr\| \\
	& \le & \Ex \Bigl\| \sum_{i=1}^n W_i \Nb_i - \bar{w} \sum_{i=1}^n \Ex \Nb_i \Bigr\|
		+ 2 \bar{v} L .
\eea
Finally,
\bea
	\Bigl( \Ex \Bigl\| \sum_{i=1}^n W_i \Nb_i - \bar{w} \sum_{i=1}^n \Ex \Nb_i \Bigr\|
		\Bigr)^2
	& = & \Bigl( \Ex \Bigl\| \sum_{i=1}^n \bigl( W_i \Nb_i - \Ex W_i \Nb_i \bigr) \Bigr\|
		\Bigr)^2 \\
	& \le & \mathrm{trace} \Var \Bigl( \sum_{i=1}^n W_i \Nb_i \Bigr) \\
	& = & \sum_{i,j=1}^n \mathrm{trace} \Cov(W_i\Nb_i, W_j\Nb_j) .
\eea
But
$$
	\mathrm{trace} \Var(W_i \Nb_i)
	\ \le \ \Ex(W_i^2 \|\Nb_i\|^2)
	\ = \ \Ex(W_i^2) \Ex(\|\Nb_i\|^2)
	\ \le \ K \bar{v} \Ex \|\Mb_i'\| ,
$$
and for $i \ne j$,
\bea
	\mathrm{trace} \Cov(W_i\Nb_i, W_j\Nb_j)
	& = & \Ex( W_i W_j \Nb_i^\top \Nb_j^{}) - \bar{w}^2 (\Ex \Nb_i)^\top (\Ex \Nb_j) \\
	& = & \bigl( \Ex (W_i W_j) - \bar{w}^2 \bigr)
		(\Ex \Nb_i)^\top (\Ex \Nb_j) \\
	& = & \Bigl( \frac{1}{n(n-1)} \sum_{k,\ell=1}^n 1\{k \ne \ell\} W_k W_\ell - \bar{w}^2 \Bigr)
		 (\Ex \Nb_i)^\top (\Ex \Nb_j) \\
	& = & \Bigl( \frac{n}{n-1} \bar{w}^2 - \frac{1}{n(n-1)} \sum_{k=1}^n W_k^2 - \bar{w}^2 \Bigr)
		 (\Ex \Nb_i)^\top (\Ex \Nb_j) \\
	& = & \frac{-1}{n-1} \, \Var(W_1) (\Ex \Nb_i)^\top (\Ex \Nb_j) .
\eea
Consequently,
\bea
	\lefteqn{ \Bigl(
		\Ex \Bigl\| \sum_{i=1}^n W_i \Nb_i - \bar{w} \sum_{i=1}^n \Ex \Nb_i \Bigr\|
		\Bigr)^2 }Ê\\
	& \le & K \bar{v} L - \frac{1}{n-1} \Var(W_1) \sum_{i,j=1}^n 1\{i \ne j\}
		(\Ex \Nb_i)^\top (\Ex \Nb_j) \\
	& = & K \bar{v} L
		+ \frac{1}{n-1} \Var(W_1) \sum_{i=1}^n \|\Ex \Nb_i\|^2
		- \frac{1}{n-1} \Var(W_1) \Bigl\| \sum_{i=1}^n \Ex \Nb_i \Bigr\|^2 \\
	& \le & K \bar{v} L + \frac{1}{n-1} \Ex(W_1^2) \sum_{i=1}^n \Ex(\|\Nb_i\|^2) \\
	& \le & \frac{n}{n-1} \, K \bar{v} L .
\eea\\[-7ex]
\end{proof}

\begin{proof}[\bf Proof of Theorem~\ref{thm:Asymptopia}]
We start with part~(a). Note first that
$$
	n^{1/2}(\check{\thb}_n^{(b)} - \thb_n^{})
	\ = \ \Gamb_{n,b}^{-1} \Zb_{n,b}^{}
$$
for $b = 0, 1, 2, \ldots, B$, where
\bea
	\Gamb_{n,b} & := & \frac{1}{n} \sum_{i=1}^n W_{ni}^{(b)} \xb_{ni}^{} \xb_{ni}^\top , \\
	\Zb_{n,b}   & := & n^{-1/2} \sum_{i=1}^n W_{ni}^{(b)} \eps_{ni}^{} \xb_{ni}^{}
\eea
with $W_{ni}^{(0)} := 1$. By Slutsky's lemma, it suffices to show that
\begin{eqnarray}
\label{eq:Asymptopia.1}
	\Gamb_{n,b} & \to_p & \Gamb \quad \text{for} \ b=0,1,\ldots,B , \\
\label{eq:Asymptopia.2}
	\bigl( \Zb_{n,b}^{} \bigr)_{b=0}^B
	& \to_{\LL} & \bigl( \Zb^{(0)} + 1\{b \ge 1\} c \Zb^{(b)} \bigr)_{b=0}^B .
\end{eqnarray}

Since $\Gamb_{n,0} = n^{-1} \sum_{i=1}^n \xb_{ni}^{} \xb_{ni}^\top$ converges to $\Gamb$ by assumption~(D.1), claim~\eqref{eq:Asymptopia.1} is equivalent to
\begin{equation}
\label{eq:Asymptopia.1'}
	\frac{1}{n} \sum_{i=1}^n W_{ni}^{} \xb_{ni}^{} \xb_{ni}^\top
	\ \to_p \ \Gamb .
\end{equation}

Concerning claim~\eqref{eq:Asymptopia.2}, note that $\bigl( \Zb_{n,b} \bigr)_{b=0}^B = n^{-1/2} \sum_{i=1}^n \eps_{ni} \Ab_{ni}$ with
$$
	\Ab_{ni}^{} \ := \ \bigl( W_{ni}^{(b)} \xb_{ni}^{} \bigr)_{b=0}^B .
$$
If we condition on the weight vectors $\Wb_n^{(b)}$, the $\Ab_{ni}$ are fixed vectors in $\R^{q(B+1)}$. Thus the multivariate version of Lindeberg's central limit theorem implies claim \eqref{eq:Asymptopia.2}, provided that the following two conditions are satisfied:
\begin{equation}
\label{eq:Asymptopia.2a}
	\frac{1}{n} \sum_{i=1}^n \Var(\eps_{ni}) \Ab_{ni}^{} \Ab_{ni}^\top
	\ \to_p \ \Var \Bigl( \Bigl( \Zb^{(0)} + 1\{b \ge 1\} c \Zb^{(b)} \Bigr)_{b=0}^B \Bigr) ,
\end{equation}
\begin{equation}
\label{eq:Asymptopia.2b}
	\Ex_* \frac{1}{n} \sum_{i=1}^n \eps_{ni}^2 \|\Ab_{ni}\|^2 
		1 \bigl\{ \eps_{ni}^2 \|\Ab_{ni}\|^2 > n \delta \bigr\}
	\ \to_p \ 0
	\quad\text{for any fixed} \ \delta > 0 .
\end{equation}
where $\Ex_*$ denotes conditional expectation, given the weight vectors $\Wb_n^{(b)}$. Due to the special structure of $\Ab_{ni}$, and in view of (D.1) and (D.3), the two claims \eqref{eq:Asymptopia.2a} and \eqref{eq:Asymptopia.2b} are easily shown to be equivalent to the following four statements:
\begin{eqnarray}
\label{eq:Asymptopia.2a'}
	\Ex_* \frac{1}{n} \sum_{i=1}^n W_{ni} \eps_{ni}^2 \xb_{ni}^{} \xb_{ni}^\top
	& \to_p & \Gamb_\eps , \\
\label{eq:Asymptopia.2a''}
	\Ex_* \frac{1}{n} \sum_{i=1}^n W_{ni}^2 \eps_{ni}^2 \xb_{ni}^{} \xb_{ni}^\top
	& \to_p & (1 + c^2) \Gamb_\eps , \\
\label{eq:Asymptopia.2a'''}
	\Ex_* \frac{1}{n} \sum_{i=1}^n W_{ni}^{(1)} W_{ni}^{(2)} \eps_{ni}^2 \xb_{ni}^{} \xb_{ni}^\top
	& \to_p & \Gamb_\eps ,
\end{eqnarray}
and
\begin{equation}
\label{eq:Asymptopia.2b'}
	\Ex_* \frac{1}{n} \sum_{i=1}^n W_{ni}^2 \eps_{ni}^2 \|\xb_{ni}\|^2 
		1 \bigl\{ W_{ni}^2 \eps_{ni}^2 \|\xb_{ni}\|^2 > n \delta \bigr\}
	\ \to_p \ 0 ,
	\quad\text{for any fixed} \ \delta > 0 .
\end{equation}

All claims \eqref{eq:Asymptopia.1'}, \eqref{eq:Asymptopia.2a'}, \eqref{eq:Asymptopia.2a''}, \eqref{eq:Asymptopia.2a'''} involve a random matrix of the form
$$
	\Ex_* \sum_{i=1}^n V_{ni} \Mb_{ni}
$$
where $V_{ni}$ denotes $W_{ni}$, $W_{ni}^2$ or $W_{ni}^{(1)} W_{ni}^{(2)}$ and $\Mb_{ni}$ stands for $n^{-1} \xb_{ni}^{} \xb_{ni}^\top$ or $n^{-1} \eps_{ni}^2 \xb_{ni}^{} \xb_{ni}^\top$. Let $\Ex_o$ denote conditional expectation, conditional on the order statistics of each weight vector $\Wb_n^{(b)}$. That means, we consider each $\Wb_n^{(b)}$ as a random permutation of a fixed weight vector. With $\bar{V}_n := n^{-1} \sum_{i=1}^n V_{ni}$ and treating matrices in $\R^{q\times q}$ as vectors in $\R^{q^2}$, it follows from Lemma~\ref{lem:Auxiliary} that for arbitrary $K \ge 1$,
\bea
	\Ex_o \Bigl\| \Ex_* \sum_{i=1}^n V_{ni} \Mb_{ni} - \bar{V}_n \sum_{i=1}^n \Ex \Mb_{ni} \Bigl\|
	& \le & \Ex_o \Ex_*
		\Bigl\| \sum_{i=1}^n V_{ni} \Mb_{ni} - \bar{V}_n \sum_{i=1}^n \Ex \Mb_{ni} \Bigl\| \\
	& = & \Ex_o \Bigl\| \sum_{i=1}^n V_{ni} \Mb_{ni} - \bar{V}_n \sum_{i=1}^n \Ex \Mb_{ni} \Bigl\| \\
	& \le & 2 R_n(K) S_n + 2 \bar{V}_n L_n
		+ \Bigl( \frac{n}{n-1} K \bar{V}_n L_n \Bigr)^{1/2} ,
\eea
where
\bea
	R_n(K) & := & n^{-1} \sum_{i=1}^n V_{ni} 1\{V_{ni} > K\} , \\
	S_n & := & \sum_{i=1}^n \Ex \|\Mb_{ni}\|
		\ \le \ \frac{1}{n} \sum_{i=1}^n (1 + \Var(\eps_{ni})) \|\xb_{ni}\|^2
		\ \to \ \mathrm{trace}(\Gamb + \Gamb_\eps) , \\
	L_n & := & \sum_{i=1}^n \Ex \|\Mb_{ni}\| \min(\|\Mb_{ni}\|,1)
		\ \le \ \Ex \sum_{i=1}^n L_{ni} \min(L_{ni},1)
		\ \to \ 0 ,
\eea
according to (D.1-3). Note also that
$$
	\sum_{i=1}^n \Ex \Mb_{ni} \ = \ \begin{cases}
		\displaystyle
		\frac{1}{n} \sum_{i=1}^n \xb_{ni}^{} \xb_{ni}^\top
		\ \to \ \Gamb
		& \text{if} \ \Mb_{ni} = n^{-1} \xb_{ni}^{} \xb_{ni}^\top , \\
		\displaystyle
		\frac{1}{n} \sum_{i=1}^n \Var(\eps_{ni}) \xb_{ni}^{} \xb_{ni}^\top
		\ \to \ \Gamb_\eps
		& \text{if} \ \Mb_{ni} = n^{-1} \eps_{ni}^2 \xb_{ni}^{} \xb_{ni}^\top .
	\end{cases}
$$
Thus it remains to verify that
\begin{equation}
\label{eq:RK}
	R_n(K) \ \to_p \ 0	\quad\text{as} \ n, K \to \infty ,	
\end{equation}
and that
$$
	\bar{V}_n \ \to_p \ \begin{cases}
		1 & \text{if} \ V_{ni} = W_{ni} \ \text{or} \ V_{ni} = W_{ni}^{(1)} W_{ni}^{(2)} , \\
		1 + c^2 & \text{if} \ V_{ni} = W_{ni}^2 .
	\end{cases}
$$

\noindent
Case 1: $V_{ni} = W_{ni}$. \ It follows from the Cauchy--Schwarz inequality that
$$
	R_n(K)^2 \ \le \ \frac{1}{n} \sum_{i=1}^n W_{ni}^2 1\{W_{ni} > K\} ,
$$
so \eqref{eq:RK} follows from (W.3). Moreover, $\bar{V}_n \equiv 1$.

\noindent
Case 2: $V_{ni} = W_{ni}^{(1)} W_{ni}^{(2)}$. \ Condition~\eqref{eq:RK} follows from the previous consideration and
\bea
	\Ex_o R_n(K)
	& \le & \Ex_o \frac{1}{n} \sum_{i=1}^n W_{ni}^{(1)} W_{ni}^{(2)}
		\bigl( 1 \{W_{ni}^{(1)} > K^{1/2}\} + 1\{W_{ni}^{(2)} > K^{1/2}\} \bigr) \\
	& = & \frac{1}{n} \sum_{i=1}^n W_{ni}^{(1)} 1 \{W_{ni}^{(1)} > K^{1/2}\}
		+ \frac{1}{n} \sum_{i=1}^n W_{ni}^{(2)} 1\{W_{ni}^{(2)} > K^{1/2}\} .
\eea
Furthermore, elementary calculations reveal that
\bea
	\Ex_o \bar{V}_n
	& = & 1 , \\
	\Var_o(\bar{V}_n)
	& = & \frac{1}{n^2(n-1)}
			\sum_{i=1}^n (W_{ni}^{(1)} - 1)^2 \sum_{j=1}^n (W_{nj}^{(2)} - 1)^2
		\ = \ O_p(n^{-1})
\eea
by (W.1-2), so $\bar{V}_n \to_p 1$.

\noindent
Case 3: $V_{ni} = W_{ni}^2$. \ Here condition~\eqref{eq:RK} is just (W.3), while
$$
	\bar{V}_n \ = \ \frac{1}{n} \sum_{i=1}^n W_{ni}^2
	\ = \ 1 + \frac{1}{n} \sum_{i=1}^n (W_{ni} - 1)^2
	\ \to_p \ 1 + c^2 
$$
by (W.1-2).

Concerning part~(b), for any fixed $\delta > 0$, it follows from part~(a) and the continuous mapping and portmanteau theorems that for any fixed integer $B \ge 1$,
\bea
	\lefteqn{ \limsup \,
	\Pr \Bigl( \min_{b=1,2,\ldots,B_n} \, n^{1/2} \bigl\| \check{\thb}_n^{(b)} - \thb_n \bigr\|
		\ge \delta \Bigr) } \\
	& \le & 	\limsup \,
		\Pr \Bigl( \min_{b=1,2,\ldots,B} \, n^{1/2} \bigl\| \check{\thb}_n^{(b)} - \thb_n \bigr\|
			\ge \delta \Bigr) \\
	& \le &
		\Pr \Bigl( \min_{b=1,2,\ldots,B} \,
			\bigl\| \Gamb^{-1}(\Zb^{(0)} + c \Zb^{(b)}) \bigr\|
			\ge \delta \Bigr) .
\eea
But Theorem~\ref{thm:Limit experiment} implies that the right hand side tends to zero as $B \to \infty$.
\end{proof}

\begin{proof}[\bf Proof of Theorem~\ref{thm:Limit experiment}]
We write $\B(\bs{a},r)$ for the closed ball in $\R^q$ with center at $\bs{a}$ and radius $r$. Set $\gamma_{B,\zb} := \alpha_q \bigl( f(\zb) B \bigr)^{1/q}$. Note that $\alpha_q^{\,q}$ is the $q$-dimensional volume of $\B(\bs{0},1)$. Thus
\bea
	\Pr \Bigl( \gamma_{B,\zb} \min_{b=1,2,\ldots,B} \bigl\| \Zb_b - \zb \bigr\|
		> x \Bigr)
	& = & \Pr \Bigl( \Zb_b \not\in \B \bigl(\zb, x/\gamma_{B,\zb} \bigr) \
		\text{for} \ b = 1,\ldots,B \Bigr) \\
	& = & \Pr \Bigl( \Zb_1 \not\in \B \bigl(\zb, x/\gamma_{B,\zb} \bigr) \Bigr)^B \\
	& = & \biggl( 1 -
		\int_{\B(\zb,\,x/\gamma_{B,\zb})} f(\bs{y}) \, d \bs{y} \biggr)^B \\
	& = & \Bigl( 1 - \alpha_q^{\,q} (x/\gamma_{B,\zb})^q \bigl( f(\zb) + o(1) \bigr) \Bigr)^B \\
	& = & \bigl( 1 - B^{-1} (x^q + o(1)) \bigr)^B \\
	& \longrightarrow & \exp(-x^q)
\eea
as $B \to \infty$.
\end{proof}

\begin{proof}[\bf Proof of Corollary~\ref{cor:Limit experiment}]
If we condition on $\Zb_0$, it follows from Theorem~\ref{thm:Limit experiment} that
$$
	B^{1/q} \min_{b=1,\ldots,B} \|\Zb_0 + c \Zb_b\|
	\ = \ c B^{1/q} \min_{b=1,\ldots,B} \bigl\|\Zb_b - (- \Zb_0/c) \bigr\|
$$
converges in distribution to
$$
	c \alpha_q^{-1} f(\Zb_0/c)^{-1/q} W ,
$$
where $W \sim \mathrm{Weibull}(q)$ is assumed to be independent from $\Zb_0$. But
\bea
	c \alpha_q^{-1} f(\Zb_0/c)^{-1/q}
	& = & c \pi^{-1/2} \Gamma(q/2 + 1)^{1/q} \, (2\pi)^{1/2} \det(\Sigb)^{1/(2q)}
		\exp \Bigl( \frac{\Zb_0^\top \Sigb^{-1} \Zb_0^{}}{2 c^2 q} \Bigr) \\
	& = & \beta_q \, \det(\Sigb)^{1/(2q)} \,
		c \, \exp \Bigl( \frac{\Zb_0^\top \Sigb^{-1} \Zb_0^{}}{2 c^2 q} \Bigr) ,
\eea
and $\Zb_0^\top \Sigb^{-1} \Zb_0^{} \sim \chi_q^2$.
\end{proof}

%=============
\end{document}